\documentclass[11pt]{article}

\usepackage{lineno}
\usepackage{amssymb,amsthm,amsmath}
\usepackage{graphicx}
\usepackage[small]{caption}
\usepackage{subcaption}
\usepackage{epsfig}
\usepackage{amsfonts}
\usepackage{booktabs}

\usepackage[utf8]{inputenc}
\usepackage{fullpage}
\usepackage{framed}

\usepackage{enumerate}
\usepackage{url}
\usepackage{hyperref}
\usepackage{multirow}

\usepackage[noline,linesnumbered]{algorithm2e}
\SetArgSty{textrm}
\let\oldnl\nl
\newcommand{\nonl}{\renewcommand{\nl}{\let\nl\oldnl}}
\makeatletter
\def\TitleOfAlgo{\@ifnextchar({\@TitleOfAlgoAndComment}{\@TitleOfAlgoNoComment}}
\def\@TitleOfAlgoAndComment(#1)#2{\nonl\hspace*{-1.5em}#2 #1\;}
\def\@TitleOfAlgoNoComment#1{\nonl\hspace*{-1.5em}#1\;}
\makeatother
\DontPrintSemicolon

\newcommand*\patchAmsMathEnvironmentForLineno[1]{
  \expandafter\let\csname old#1\expandafter\endcsname\csname #1\endcsname
  \expandafter\let\csname oldend#1\expandafter\endcsname\csname end#1\endcsname
  \renewenvironment{#1}
  {\linenomath\csname old#1\endcsname}
  {\csname oldend#1\endcsname\endlinenomath}}
  \newcommand*\patchBothAmsMathEnvironmentsForLineno[1]{
  \patchAmsMathEnvironmentForLineno{#1}
  \patchAmsMathEnvironmentForLineno{#1*}}
  \AtBeginDocument{
  \patchBothAmsMathEnvironmentsForLineno{equation}
  \patchBothAmsMathEnvironmentsForLineno{align}
  \patchBothAmsMathEnvironmentsForLineno{flalign}
  \patchBothAmsMathEnvironmentsForLineno{alignat}
  \patchBothAmsMathEnvironmentsForLineno{gather}
  \patchBothAmsMathEnvironmentsForLineno{multline}
}

\newtheorem{theorem}{Theorem}
\newtheorem{lemma}{Lemma}

\newtheorem{corollary}{Corollary}

\newcommand{\etal}{{et~al.}}
\newcommand{\ie}{{i.e.}}
\newcommand{\eg}{{e.g.}}

\newcommand{\later}[1]{}
\newcommand{\old}[1]{}

\title{\textsc{Finding Small Complete Subgraphs Efficiently}
\footnote{
  A preliminary version of this paper appeared in the
  \emph{Proceedings of the 34th International Workshop on Combinatorial Algorithms}
  (IWOCA 2023),  LNCS~13889, Springer, Cham, Switzerland, 2023, pp.~185--196.}
}

\author{
Ke Chen\thanks{%
Department of Computer Science, The Pennsylvania State University, PA, USA.
Email~\texttt{kxc5915@psu.edu}}
 \and
Adrian Dumitrescu\thanks{%
Algoresearch L.L.C., Milwaukee, WI, USA, and 
Research Institute of the University of Bucharest, Romania, and 
Alfr\'ed R\'enyi Institute of Mathematics, Budapest, Hungary. 
Email~\texttt{ad.dumitrescu@algoresearch.org}}
\and
Andrzej~Lingas\thanks{%
  Department of Computer Science, Lund University,
  Sweden.
Email~\texttt{Andrzej.Lingas@cs.lth.se}}
}

\begin{document}

\maketitle

\begin{abstract}
  (I) We revisit the algorithmic problem of finding all triangles in a graph $G=(V,E)$ with $n$ vertices
  and $m$ edges.  According to a result of Chiba and Nishizeki (1985), 
  this task can be achieved by a combinatorial algorithm running in
$O(m \alpha) = O(m^{3/2})$ time, where $\alpha= \alpha(G)$ is the graph arboricity.
We provide a new very simple combinatorial algorithm for finding all triangles in a graph
and show that is amenable to the same running time analysis. 
We derive these worst-case bounds from first principles and with very simple proofs
that do not rely on classic results due to Nash-Williams from the 1960s.
Our experimental results show that our simple algorithm for triangle
listing is substantially faster in practice than that of Chiba and
Nishizeki on all examples of real-world graphs we tried.

\smallskip
(II) We extend our arguments to the problem of finding all small complete subgraphs
of a given fixed size.  We show that the dependency on $m$ and $\alpha$ in the running time
$O(\alpha^{\ell-2} \cdot m)$  of the algorithm of Chiba and Nishizeki for listing all copies of $K_\ell$,
where $\ell \geq 3$, is asymptotically tight.

\smallskip
(III) We give improved arboricity-sensitive running times for counting and/or detection of copies of $K_\ell$,
for small $\ell \geq 4$. A key ingredient in our algorithms is, once again, the algorithm of Chiba and Nishizeki.
Our new algorithms are faster than all previous algorithms in certain high-range arboricity intervals
for every $\ell \geq 7$.


\end{abstract}

\section{Introduction} \label{sec:intro}

The problem of deciding whether a given graph $G=(V,E)$ with $n$ vertices
and $m$ edges contains a complete subgraph on $k$ vertices
is among the most natural and easily stated algorithmic graph problems.
If the subgraph size $k$ is part of the input, this is the \textsc{Clique} problem
which is NP-complete~\cite{GJ79}.
For every fixed $k$, determining whether a given graph $G=(V,E)$ contains a complete subgraph
on $k$ vertices can be accomplished by a brute-force algorithm running in $O(n^k)$ time.

For $k=3$, deciding whether a graph contains a triangle and finding one if it does,
or counting all triangles in a graph,
can be done in $O(n^\omega)$ time by the algorithm of Itai and Rodeh~\cite{IR78},
where $\omega$ is the exponent of matrix multiplication~\cite{AV21,CW90},
and $\omega < 2.372$ is the best known bound~\cite{DWZ23,VXXZ24}.
The algorithm compares $M$ and $N=M^2$, where $M$ is the graph adjacency matrix;
if there is a pair of indexes $i,j$ such that $M[i,j]=N[i,j]=1$, then there is a triangle based on edge $ij$.
An equivalent characterization is: $G$ contains a triangle iff $M^3$ has a non-zero entry on its main diagonal.
See~\cite[Ch.~10]{Mat10} for a short exposition of this elegant method.
Alternatively, this task can be done in $O(m^{2\omega/(\omega+1)}) =O(m^{1.407})$ time
by the algorithm of Alon, Yuster, and Zwick~\cite{AYZ97}, which deals with vertices of high and low degree
separately. Itai and Rodeh~\cite{IR78} and also Papadimitriou and Yannakakis~\cite{PY81}
as well as Chiba and Nishizeki~\cite{CN85} showed that triangles in planar graphs can be found in $O(n)$ time. 

For $k=4$, deciding whether a graph contains a $K_4$ and finding one if it does
(or counting all $K_4$'s in a graph) can be done in $O(n^{\omega+1}) = O(n^{3.373})$ time by the algorithm
of Alon, Yuster, and Zwick~\cite{AYZ97},
in $O(n^{3.252})$ time by the algorithm of Eisenbrand and Grandoni~\cite{EG04}, and in
$O(m^{(\omega+1)/2})=O(m^{1.686})$ time by the algorithm of Kloks, Kratsch, and M{\"{u}}ller~\cite{KKM00}.

In contrast to the problem of detecting the existence of subgraphs of a certain kind,
the analogous  problem of listing \emph{all} such subgraphs has usually higher complexity,
as expected. For example, finding all triangles in a given graph
(each triangle appears in the output list) can be accomplished
in $O(m^{3/2})$ time and with $O(n^2)$ space by an extended
version of the algorithm of Itai and Rodeh~\cite{IR78}. 
Bar-Yehuda and Even~\cite{BE82} improved the space complexity of the algorithm
from $O(n^2)$ to $O(n)$ by avoiding the use of the adjacency matrix.

Chiba and Nishizeki~\cite{CN85} further refined the time complexity
in terms of graph arboricity, which is the minimum number of edge-disjoint
forests into which its edges can be partitioned. Their algorithm lists all triangles in a graph
in $O(m \alpha)$ time, where $\alpha$ is the arboricity of the graph.
They also showed that $\alpha =O(\sqrt{m})$; consequently, the running time is $O(m^{3/2})$.
Since there are graphs $G$ with $\alpha(G)= \Theta(m^{1/2})$,
this does not improve the worst-case dependence on $m$.
Moreover, since there are graphs with $\Theta(m^{3/2})$ triangles, see, \eg, \cite{Ri02},
this dependence cannot be improved.
For every fixed $\ell \geq 3$, the same authors gave an algorithm for
listing all copies of $K_\ell$ in $O(\alpha^{\ell-2} \cdot m) = O(m^{\ell/2})$ time. 
If $G$ is planar,  $\alpha(G) \leq 3$ (see~\cite[p.~124]{Har72}),
so their algorithm runs in $O(m)=O(n)$ (\ie, linear) time on planar graphs.

We distinguish several variants of the general problem of finding triangles in a given
undirected graph $G=(V,E)$:
(i)~the triangle \emph{detection} problem is that of deciding
if $G$ is triangle-free (resp., finding a triangle in $G$ otherwise);
(ii)~the triangle \emph{counting} problem is that of determining the total number 
  of triangles in $G$;
(iii)~the triangle \emph{listing} problem is that of listing (finding, reporting) \emph{all} triangles in $G$,
  with each triangle appearing in the output list. 
Obviously, any algorithm for listing all triangles can be easily transformed into one for
triangle detection or into one for listing a specified number of triangles,
by stopping after the required number of triangles have been output.

\paragraph{Our results.}
We obtain the following results for the problem of finding small complete subgraphs of a given size.
Each of our algorithms for deciding whether a graph contains a given complete subgraph (item (iii) below)
also finds a suitable witness, if it exists.

\begin{enumerate} [(i)] \itemsep 1pt

\item We provide a new combinatorial algorithm for finding all triangles in a graph
  running in $O(m \alpha) = O(m^{3/2})$ time, where $\alpha= \alpha(G)$ is the graph arboricity
   (Algorithm \textsc{Hybrid} in Section~\ref{sec:hybrid}).
We derive these worst-case bounds from first principles and with very simple proofs
that do not rely on classic results due to Nash-Williams from the 1960s.
We present also an experimental comparison of our new simple algorithm
for finding all triangles with that of Chiba and Nishizeki. It shows
that our algorithm is substantially faster in practice on all
considered examples of real-world graphs.

\item  For every $n$, $b \leq n/2$ and a fixed $\ell \geq 3$,
  there exists a graph $G$ of order $n$ with $m$ edges and $\alpha(G) \leq b$
  that has $\Omega(\alpha^{\ell-2} \cdot m)$ copies of $K_\ell$
  (Lemma~\ref{lem:ell} in Section~\ref{sec:constructions}).
  As such, the dependency on $m$ and $\alpha$, in the running time
  $O(\alpha^{\ell-2} \cdot m)$ for listing all copies of $K_\ell$,
  is asymptotically tight. 

\item We give improved arboricity-sensitive running times for counting and/or detection of copies of $K_\ell$,
for small $\ell \geq 4$ (Section~\ref{sec:small}).
A key ingredient in our algorithms is the algorithm of Chiba and Nishizeki.
Our new algorithms beat all previous algorithms in certain high-range arboricity intervals
for every $\ell \geq 7$. We also provide up-to-date running times based on rectangular matrix
multiplication times.

\end{enumerate}

\paragraph{Preliminaries and notation.} 
Let $G=(V,E)$ be an undirected graph. The \emph{neighborhood} of a vertex $v \in V$ is
the set $N(v) =\{w \ : \  (v,w) \in E\}$ of all adjacent vertices, and its cardinality 
$\deg(v)=|N(v)|$ is the \emph{degree} of $v$ in $G$.  
A \emph{clique} in a graph $G=(V,E)$ is a subset $C \subset V$ of vertices,
each pair of which is connected by an edge in $E$. The \textsc{Clique}
problem is to find a clique of maximum size in $G$. 
An \emph{independent set} of a graph $G=(V,E)$ is a subset $I \subset V$ of vertices
such that no two of them are adjacent in $G$. 

Unless specified otherwise:
(i)~all algorithms discussed are assumed to be deterministic;
(ii)~all graphs mentioned are undirected; 
(iii)~all logarithms are in base~$2$. 

It should be noted that the asymptotically fastest algorithms for matrix multiplication
have mainly a theoretical importance and are otherwise largely impractical~\cite[Ch.~10.2.4]{Man89}.
As such, any algorithm that employs fast matrix multiplication (FMM) inherits this undesirable feature.
Whereas the notion of ``combinatorial'' algorithm does not have a formal or precise definition,
such an algorithm is usually simpler to implement and practically efficient, see~\cite{VW18}.
Combinatorial algorithms are sometimes theoretically less efficient than non-combinatorial ones
but practically more efficient since they tend to be simpler to implement.

\paragraph{Graph parameters.}
For a graph $G$, its \emph{arboricity} $\alpha(G)$ is the minimum number of edge-disjoint
forests into which $G$ can be decomposed~\cite{CN85}. 
For instance, it is known and easy to show that $\alpha(K_n)= \lceil n/2 \rceil$.
A characterization of arboricity is provided by the following classic
result~\cite{NW61,NW64,Tu61}; see also~\cite[p.~60]{Die97}.

\begin{theorem} {\rm (Nash-Williams 1964; Tutte 1961)} \label{thm:k-forests}
A multigraph $G=(V,E)$ can be partitioned into at most $k$ forests if and only if
every set $U \subseteq V$ induces at most $k(|U|-1)$ edges. 
\end{theorem}

The \emph{degeneracy} $d(G)$ of an undirected graph $G$ is the smallest number $d$
for which there exists an \emph{acyclic orientation} of $G$ in which all the
out-degrees are at most $d$. The degeneracy of a graph is linearly related
to its arboricity, \ie, $\alpha(G)=\Theta(d(G))$; more precisely
$ \alpha(G) \leq d(G) \leq 2\alpha(G)-1$; see~\cite{AYZ97,Epp94,LW70,ZN99}.

It is worth noting that high arboricity is not an indication for the graph having many (or any) triangles.
Indeed, the complete bipartite graph $K_{\lfloor n/2 \rfloor,\lceil n/2 \rceil}$, or any of its dense subgraphs has
arboricity $\Theta(n)$ but no triangles.

\subsection{Triangle Finding Algorithms} \label{sec:alg}

To place our algorithm in Section~\ref{sec:hybrid} in a proper context,
we first present a summary of previous work in the area of triangle finding and enumeration.

\paragraph{The algorithms of Itai and Rodeh (1978).}
The authors gave three methods for finding triangles. 
Initially intended for triangle detection, the first algorithm~\cite{IR78} runs in $O(m^{3/2})$ time.
It can be extended to list all triangles within the same overall time and works as follows:

\smallskip
\begin{algorithm}[H]
  \DontPrintSemicolon
  \TitleOfAlgo{\textsc{Itai-Rodeh}$(G)$}
%  \TitleOfAlgo{\textsc{List-Triangles}$(G)$}
  \KwIn{an undirected graph $G=(V,E)$}
  Find a spanning tree for each connected component of $G$\;
  List all triangles containing at least one tree edge\;
  Delete the tree edges from $G$ and go to Step 1\;
\end{algorithm}
\smallskip

The second is a randomized algorithm that checks whether there is an edge contained in a triangle. 
It runs in $O(mn)$ worst-case time and $O(n^{5/3})$ expected time. 
The third algorithm relies on Boolean matrix multiplication and runs in $O(n^\omega)$ time.
(Boolean matrix multiplication can be reduced to standard matrix multiplication;
see~\cite[Ch.~10.2.4]{Man89} for details.)  

\paragraph{The algorithm of Chiba and Nishizeki (1985).}
The algorithm uses a vertex-iterator approach for listing all triangles in $G$.
It relies on the observation that each triangle containing a vertex $v$ corresponds
to an edge joining two neighbors of $v$.
The graph is represented with doubly-linked adjacency lists and mutual references
between the two stubs of an edge ensure that each deletion takes constant time.  
A more compact version described by Ortmann and Brandes~\cite{OB14} is given below.

\smallskip
\begin{algorithm}[H]
  \DontPrintSemicolon
  \TitleOfAlgo{\textsc{Chiba-Nishizeki}$(G)$}
  \KwIn{an undirected graph $G=(V,E)$}
  Sort vertices such that $\deg(v_1) \geq \deg(v_2) \geq \ldots \deg(v_n)$\;
  \For{$u=v_1,v_2,\ldots,v_{n-2}$}{
    \lForEach{vertex $v \in N(u)$}{mark $v$}
    \ForEach{vertex $v \in N(u)$}{
      \ForEach{vertex $w \in N(v)$}{
        \lIf{$w$ is marked}{output triangle $uvw$}}
    unmark $v$
  }
    $G \gets G - u$
    }
  \end{algorithm}
\smallskip

The authors~\cite{CN85} showed that their algorithm runs in $O(m \alpha)$ time.
As a corollary, the number of triangles is $O(m \alpha)$ as well.
The $O(m^{3/2})$ upper bound on the number of triangles in a graph
is likely older than these references indicate. In any case, other proofs
are worth mentioning~\cite{ELRS17,KMPT12}, including algebraic ones~\cite{Ri02}.
We derive yet another one in Section~\ref{sec:hybrid}.

\begin{corollary} {\rm (\cite{IR78,CN85})} \label{cor:no-of-t}
  For any graph $G$ of order $n$ with $m$ edges and arboricity $\alpha$,
  $G$ contains $O(m \alpha)=O(m^{3/2})$ triangles.
\end{corollary}

Ortmann and Brandes~\cite{OB14} gave a survey of other approaches, including edge-iterators.
Algorithms in this category iterate over all edges and intersect the neighborhoods of
the endpoints of each edge. 
A straightforward neighborhood merge requires $O(\deg(u) + \deg(v))$ time per edge $uv$,
but this is not good enough to list all triangles in $O(m^{3/2})$ time. 
Two variants developed by Shanks~\cite{Sch07} use $O(m)$ extra space to represent
neighborhoods in hash sets and obtain the intersection in  $O(\min(\deg(u),\deg(v)))$ time,
which suffices for listing all triangles in $O(m^{3/2})$ time.

\paragraph{The algorithm of Alon, Yuster and Zwick (1997).}
The authors showed that deciding whether a graph contains a triangle and finding
one if it does (or counting all triangles in a graph)
can be done in $O(m^{2\omega/(\omega+1)}) =O(m^{1.407})$ time~\cite{AYZ97}.
The idea is to find separately triangles for which at least one vertex has low degree (for an
appropriately set threshold) and triangles  whose all three vertices have high degree.
Triangles of the latter type are handled by applying matrix multiplication to a smaller
subgraph.

\paragraph{Recent algorithms.}
Bj{\"{o}}rklund~\etal~\cite{BPWZ14} obtained output-sensitive algorithms
for finding (and listing) all triangles in a graph;
their algorithms are tailored for dense and sparse graphs.
Several approaches~\cite{EGMT17,KPP15}
  provide asymptotic improvements by taking advantage of the bit-level
  parallelism offered by the word-RAM model.
  Although they do not appear to be very practical,
  these methods asymptotically improve on the $O(m \alpha +k)$ running time of the
  earlier algorithms (in particular that in~\cite{CN85}).
  If the number of triangles is small, Zechner and Lingas~\cite{ZL14} showed how to list
  all triangles in $O(n^\omega)$ time.

\section{A Simple Hybrid Algorithm for Listing All Triangles} \label{sec:hybrid}

In this section we present a new algorithm for listing all triangles. 
While its general idea is not new, the specific hybrid data structure therein does not appear
to have been previously considered. Using both an adjacency list representation 
and an adjacency matrix representation of the graph yields a very time-efficient
neighborhood merge (intersection) procedure.
Let $V=\{1,2,\ldots,n\}$ and let $M$ be the adjacency matrix of $G$. 
A triangle $ijk$ with $i<j<k$ is reported when edge $ij$ is processed, and so each  
triangle is reported exactly once.
For each edge $ij$, the algorithm scans the adjacency list
of the endpoint of lower degree (among $i$ and $j$)
and for each neighbor it checks for the needed entry in the adjacency matrix $M$
corresponding to the third triangle edge. In the pseudocode below, $ADJ(v)$ denotes the adjacency list of vertex $v$. 

\smallskip
\begin{algorithm}[H]
  \DontPrintSemicolon
  \TitleOfAlgo{\textsc{Hybrid}$(G)$}
%  \TitleOfAlgo{\textsc{List-All-Triangles}$(G)$}
  \KwIn{an undirected graph $G=(V,E)$}
  \ForEach{edge $ij \in E$, ($i<j$)}{
    \lIf{$\deg(i) \leq \deg(j)$}{$x\gets i$, $y\gets j$}    
      \lElse{$x\gets j$, $y\gets i$}
    \ForEach{$k \in ADJ(x)$}{
      \lIf{$j<k$ and $M(y,k)=1$}{report triangle $ijk$}
    }
  }
\end{algorithm}
\smallskip

We subsequently show that the algorithm runs in time $O(m \alpha)$.
Whereas the space used is quadratic,
the hybrid algorithm appears to win by running time and simplicity.
In particular, no additional data structures nor hashing are used, no sorting
(by degree or otherwise) and no doubly linked lists are needed, etc.
Similar algorithms that use hashing in order to run in linear space are $6$ to $8$ times slower
than their ``unoptimized'' counterparts on most graphs with about $5 \times 10^6$ vertices
in the study in~\cite{Sch07,SW05}.
``The two algorithms \emph{edge-iterator-hashed} and \emph{forward-hashed},
which use hashed data structures for every node, perform badly'' on such graphs,
the authors write~\cite[p.~8]{SW05}.

\subsection{The Analysis} \label{subsec:analysis}

Define the following function $f$ on the edges of $G$. For $uv=e \in E$, let

\begin{equation} \label{eq:f}
f(e) = \min(\deg(u),\deg(v))   \text{ and } F(G) = \sum_{e \in E} f(e).
\end{equation}

There are at most $f(e)$ triangles based on edge $e=ij$ and overall at most 
$F(G)$ triangles in $G$. The runtime of the algorithm is proportional to this quantity;
the space used is quadratic.
A short and elegant decomposition argument by Chiba and Nishizeki~\cite[Lemma~2]{CN85}
shows that
\begin{equation} \label{eq:decomposition}
  F(G) \leq 2 m \alpha,
\end{equation}
thus our algorithm runs in time $O(m \alpha)$. 
The same analysis applies to the algorithm of Chiba and Nishizeki~\cite{CN85}.
By Theorem~\ref{thm:k-forests}, the above authors deduced that 
$\alpha(G) \leq \lceil (2m+n)^{1/2} /2 \rceil$,
which implies that  $\alpha(G) =O(m^{1/2})$ for a connected graph $G$. 
As such, both algorithms run in $O(m^{3/2})$ time on any graph. 

Lemma~\ref{lem:F} below shows how to bypass the Nash-Williams arboricity bound (Theorem~\ref{thm:k-forests})
and deduce the $O(m^{3/2})$ upper bound for listing all triangles in a graph from first principles.

\begin{lemma} \label{lem:F}
Let $G$ be a graph on $n$ vertices with $m$ edges. Then
\[ F(G) \leq 4 m^{3/2}. \]
\end{lemma}
\begin{proof} (folklore)
There are two types of edges $uv$:
\begin{enumerate} \itemsep 1pt 
\item $\min(\deg(u),\deg(v)) \leq 2\sqrt{m}$
\item $\min(\deg(u),\deg(v)) > 2\sqrt{m}$
\end{enumerate}

There are at most $m$ edges of type 1 and each contributes at most $2\sqrt{m}$
to $F(G)$, so the total contribution from edges of this type is at most $2 m^{3/2}$. 

Each edge of type 2 connects two nodes of degree $> 2 \sqrt{m}$, and there are
at most $2m/(2 \sqrt{m})= \sqrt{m}$ such nodes. The degree of each of them thus contributes
to $F(G)$ at most $\sqrt{m}$ times and the sum of all degrees of such nodes is at most $2m$.
It follows that the total contribution from edges of type 2 is at most $2 m^{3/2}$. 

Overall, we have $F(G) \leq 4 m^{3/2}$. 
\end{proof}

\paragraph{Remarks.} Lemma~\ref{lem:F} immediately gives an upper bound of $4 m^{3/2}$
on the number of triangles in a graph.
It is in fact known~\cite{Ri02} that this number is at most $(2^{1/2}/3)\, m^{3/2}$
(and this bound is sharp).  More generally, for every fixed $\ell \geq 3$, the number of copies of $K_\ell$
is $O(\alpha^{\ell-2} \cdot m) =O(m^{\ell/2})$; see~\cite[Thm.~3]{CN85}. 

\subsection{Experimental Results} \label{subsec:experiment}
Although \textsc{Hybrid} and \textsc{Chiba-Nishizeki} have the same asymptotic running time,
we claim that our algorithm is faster in practice due to its simplicity.
To support this claim, we implemented both algorithms in C++ and compared their running time on
complete graphs $K_{1000}$ and $K_{2000}$, as well as
a collection of real-world graphs with various sizes.
The code is publicly available at \url{https://gitlab.com/kanatos92/triangle-listing}.
The graphs are taken from the SNAP datasets~\cite{LK14} with self-loops and multi-edges removed.
Graph statistics are listed in Table~\ref{tab:graphs}.

\begin{table}[!ht]
  \centering
  \begin{tabular}{|c|r|r|r|}
    \hline
    Graph & Vertices & Edges & Triangles \\
    \hline
    $K_{1000}$ & 1,000 & 499,500 & 166,167,000\\
    $K_{2000}$ & 2,000 & 1,999,000 & 1,331,334,000\\
    facebook-combined & 4,039 & 88,234 & 1,612,010\\
    musae-facebook & 22,470 & 170,823 & 794,953\\
    cit-HepPh & 34,546 & 420,877 & 1,276,868\\
    large-twitch & 168,114 & 6,797,557 & 54,148,895\\
    loc-gowalla & 196,591 & 950,327 & 2,273,138\\
    com-dblp & 317,080 & 1,049,866 & 2,224,385\\
    higgs-social-network & 456,626 & 12,508,413 & 83,023,401\\
    web-BerkStan & 685,230 & 6,649,470 & 64,690,980\\
    com-youtube & 1,134,890 & 2,987,624 & 3,056,386\\
    as-skitter & 1,696,415 & 11,095,298 & 28,769,868\\
    roadNet-CA & 1,965,206 & 2,766,607 & 120,676\\
    \hline
  \end{tabular}
  \caption{Statistics of the graphs.}
  \label{tab:graphs}
\end{table}

For each algorithm, we measure the construction time taken to build the needed graph representation,
and the listing time taken to enumerate all the triangles.
To minimize the effect of system fluctuations, I/O operations are not timed.
In particular, each input graph is preloaded into memory as a list of edges, and the triangles
are stored in a container (instead of being sent directly to output) during enumeration.
The experiments are conducted on an AWS server with an AMD EPYC 7R13 processor and 512 GB memory.
For a fair comparison, both algorithms are run on a single processor.
However, it is worth noting that the loop in \textsc{Hybrid} can easily be made to run in parallel
(by adding a single line of pragma directive from OpenMP, for example) whereas it is not
so trivial for \textsc{Chiba-Nishizeki} because the algorithm modifies its data structures
in each iteration.

Both algorithms are repeated $10$ times on each graph, the average time is summarized in Table~\ref{tab:experiment}.
Observe that, although for some graphs \textsc{Hybrid} takes more time in the construction phase,
it is the overall winner for all the inputs in the table.
%it is overall faster for all the inputs in the table.
In most cases, the running time of \textsc{Hybrid} is less than half of the running time of \textsc{Chiba-Nishizeki}.

\begin{table}[!ht]
  \centering
  \begin{tabular}{|c|rrr|rrr|}
    \hline
    \multirow{2}{*}{Graph} & \multicolumn{3}{|c|}{\textsc{Chiba-Nishizeki}} & \multicolumn{3}{|c|}{\textsc{Hybrid}}\\
    \cline{2-7}
          & Construction & Listing & Total & Construction & Listing & Total\\
    \hline
    $K_{1000}$  & $37.4$ & $11221.8$ & $11259.2$ & $\mathbf{8.0}$ & $\mathbf{5778.3}$ & $\mathbf{5786.3}$ \\
    $K_{2000}$  & $146.6$ & $120570.2$ & $120716.8$ & $\mathbf{34.6}$ & $\mathbf{46121.5}$ & $\mathbf{46156.1}$ \\
    facebook-combined  & $6.0$ & $89.9$ & $95.9$ & $\mathbf{1.7}$ & $\mathbf{54.3}$ & $\mathbf{56.0}$ \\
    musae-facebook  & $16.7$ & $104.3$ & $121.0$ & $\mathbf{13.3}$ & $\mathbf{44.5}$ & $\mathbf{57.8}$ \\
    cit-HepPh  & $45.8$ & $888.4$ & $934.2$ & $\mathbf{41.2}$ & $\mathbf{202.1}$ & $\mathbf{243.3}$ \\
    large-twitch  & $913.9$ & $48107.8$ & $49021.7$ & $\mathbf{898.3}$ & $\mathbf{18424.4}$ & $\mathbf{19322.7}$ \\
    loc-gowalla  & $\mathbf{138.3}$ & $2500.8$ & $2639.1$ & $150.5$ & $\mathbf{951.4}$ & $\mathbf{1101.9}$ \\
    com-dblp  & $\mathbf{173.4}$ & $1189.5$ & $1362.9$ & $224.1$ & $\mathbf{454.8}$ & $\mathbf{678.9}$ \\
    higgs-social-network  & $2387.1$ & $144523.7$ & $146910.8$ & $\mathbf{2363.5}$ & $\mathbf{54425.8}$ & $\mathbf{56789.3}$ \\
    web-BerkStan  & $1219.5$ & $16841.5$ & $18061.0$ & $\mathbf{923.3}$ & $\mathbf{5828.8}$ & $\mathbf{6752.1}$ \\
    com-youtube  & $\mathbf{557.6}$ & $7915.6$ & $8473.2$ & $617.7$ & $\mathbf{3106.1}$ & $\mathbf{3723.8}$ \\
    as-skitter  & $\mathbf{2532.2}$ & $47561.0$ & $50093.2$ & $2577.4$ & $\mathbf{17494.4}$ & $\mathbf{20071.8}$ \\
    roadNet-CA  & $\mathbf{690.0}$ & $1187.3$ & $1877.3$ & $1210.8$ & $\mathbf{226.9}$ & $\mathbf{1437.7}$ \\
    \hline
  \end{tabular}
  \caption{Average running time of both algorithms in milliseconds. The faster ones in each row are highlighted in bold.}
  \label{tab:experiment}
\end{table}

\section{Constructions} \label{sec:constructions}

\begin{lemma} \label{lem:m}
  For every $n\geq3$ and $3 \leq m \leq {n \choose 2}$, there exists a graph $G$ of order $n$ with $m$ edges
  that contains $\Theta(m^{3/2})$ triangles.
\end{lemma}
\begin{proof}
  Let $3 \leq x \leq n$ be the unique integer such that ${x \choose 2} \leq m < {x+1 \choose 2}$. 
  Note that $x =\Theta(\sqrt{m})$. Let $G=(V,E)$ where $V=V_1 \cup V_2$, $|V_1|=x$, and $|V_2|=n-x$.
Set $V_1$ induces a complete subgraph and set $V_2$ induces an independent set in $G$.
Add $m- {x \choose 2}$ edges  to $G$ arbitrarily so that $G$ has $m$ edges.
Let $T$ be the set of triangles in $G$. We have
\[ |T| \geq  {x \choose 3} = \Omega(m^{3/2}). \]
By Lemma~\ref{lem:F}, $G$ contains $\Theta(m^{3/2})$ triangles, as required.
\end{proof}

\begin{lemma} \label{lem:ell}
  Let $\ell \geq 3$ be a fixed integer. 
  For every $n$ and $b \leq n/2$, there exists a graph $G$ of order $n$ with $\alpha \leq b$
  that has $\Omega(\alpha^{\ell-2} \cdot m)$ copies of $K_\ell$, where $m$ is the number of edges
  in $G$ and $\alpha$ is the arboricity of $G$.
\end{lemma}
  \begin{proof}
We may assume  that $b$ is even and $b/2 \ | \ n$. Suppose that $n=(2k+1)b/2$. 
Let $G=(V,E)$, where $V=V_1 \cup V_2$, and $|V_1|=kb$, $|V_2|=b/2$. 
Set $V_1$ consists of $k$ complete subgraphs of order $b$ whereas $V_2$ is an
independent set of size $b/2$ in $G$. 
All edges in $V_1 \times V_2$ are present in $G$. See Fig.~\ref{fig:f1}.

\begin{figure}[htbp]
\centering
\includegraphics[scale=0.75]{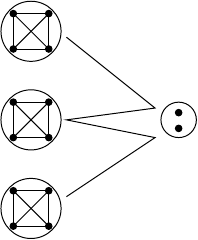}
\caption{Illustration of the graph $G$ for $k=3$, $b=4$.}
\label{fig:f1}
\end{figure}

Let $K$ be the set of complete subgraphs of order $\ell$. 
We have
\begin{align*}
m &=|E|= k {b \choose 2} + k \frac{b^2}{2}= \Theta(k \, b^2), \\
|K| &=   k {b \choose \ell} + k {b \choose \ell-1} \frac{b}{2} = \Omega(k \, b^\ell)
= \Omega(b^{\ell-2} \, m),
\end{align*}
Note that the arboricity of $G$ is bounded from above by $b$. Indeed,
$ \alpha \leq \frac{b}{2} + \frac{b}{2} =b$, 
since $\alpha(K_b)=b/2$ and the edges in $V_1 \times V_2$ can be partitioned into
$b/2$ stars centered at the vertices in $V_2$. 
\end{proof}

\section{Finding Small Complete Subgraphs Efficiently} \label{sec:small}

In this section we address the problem of detecting the presence of $K_\ell$ for a fixed $\ell \geq 4$.
We combine and refine several approaches existent in the literature of the last $40$ years
to obtain faster algorithms in general and for a large class of graphs with high arboricity.
In particular, we will use the algorithm of  Chiba and Nishizeki~\cite{CN85} 
for listing all copies of $K_\ell$ in $O(\alpha^{\ell-2} \cdot m)$ time. 
Our algorithms are formulated for the purpose of \emph{counting} 
but they can be easily adapted for the purpose of \emph{detection}.

The basic idea is as follows: since listing is generally slower than counting or detection, one can use
a \emph{hybrid}---decomposition based---approach that combines listing in Part A
with counting or detection in Part B.

Recall that $\omega < 2.372$ is the exponent of matrix multiplication~\cite{AV21,CW90,VXXZ24},
namely the infimum of numbers $\tau$ such that two $n \times n$ real matrices can be multiplied in
$O(n^\tau)$ time (operations). Similarly, let $\omega(p,q,r)$ stand for the infimum of numbers $\tau$ such that
an $n^p \times n^q$ matrix can be multiplied by an $n^q \times n^r$ matrix in $O(n^{\tau})$ time (operations). 
For simplicity and as customary (see, \eg, \cite{AYZ97,NP85}),
we write that two $n \times n$ matrices can be multiplied in $O(n^\omega)$ time,
since this does not affect our results; and similarly when multiplying rectangular matrices.

\paragraph{The extension method.}
The algorithm solves the detection (or counting) problem for a complete subgraph $K_\ell$.
Let $T(n,m,\ell)$ denote the running time of the algorithm for a graph with $n$ vertices and $m$ edges. 
Write $\ell=j_1 +j_2$, for a suitable choice $j_1,j_2 \geq 2$.
At the beginning, we run the algorithm of  Chiba and Nishizeki to form a list of subgraphs isomorphic
to $K_{j_1}$. Then, for each subgraph $G_1=(V_1,E_1)$ on the list:
(i)~we construct the subgraph $G_2$ of $G$ induced by vertices in $V \setminus V_1$
that are adjacent to all vertices in $V_1$; and 
(ii)~we count the subgraphs isomorphic to $K_{j_2}$ in $G_2$ (or find one such subgraph); 
this is a (possibly \emph{recursive}) call of the same procedure on a smaller instance.
In other words, for each subgraph isomorphic to $K_{j_1}$, we count the number of extensions to $K_\ell$.
Another formulation of this method can be found in~\cite{KL22}.

There are $O(\alpha^{j_1-2} \cdot m)$ copies of $K_{j_1}$ and they can be
found in $O(\alpha^{j_1-2} \cdot m)$ time. 
For each fixed copy of $K_{j_1}$, the time spent in $G_2$ is at most $T(n,m,j_2)$, and so the
overall time satisfies the recurrence
\[ T(n,m,\ell) = O(\alpha^{j_1-2} \cdot m \cdot T(n,m,j_2)). \]
Each copy of $K_\ell$ is generated exactly ${\ell \choose j_1}$ times and so the total count
needs to be divided by this number in the end. 

\paragraph{The triangle method and its refinement.}
The algorithm solves the detection (or counting) problem for a complete subgraph $K_\ell$.
Ne{\v{s}}et{\v{r}}il and Poljak~\cite{NP85} showed an efficient reduction of detecting and counting copies
of any complete subgraph to the aforementioned method of Itai and Rodeh~\cite{IR78} for triangle detection
and counting.  To start with, consider the detection of complete subgraphs of size $\ell =3j$.
For a given graph $G$ with $n$ vertices, construct an auxiliary graph $H$ with $O(n^j)$ vertices,
where each vertex of $H$ is a complete subgraph of order $j$ in $G$.
 Two vertices $V_1,V_2$ in $H$ are connected by an edge  if $V_1 \cap V_2 =\emptyset$ and
 all edges in $V_1 \times V_2$ are present in $G$; equivalently, $V_1 \cup V_2$ is a complete subgraph on
 $|V_1| + |V_2|$ vertices.
The detection (or counting) of triangles in $H$ yields an algorithm for the detection (or counting) of
$K_{\ell}$'s in $G$, running in $O(n^{j \omega})$ time.
For detecting complete subgraphs of size $\ell= 3j +i$, where $i \in \{1,2\}$,
the algorithm can be adapted so that it runs in $O(n^{j\omega +i})$ time.

The triangle method can be extended to larger subgraphs, \ie, if $\ell= k j +i$,
for some $k \geq 4$, and $i \in \{0,1,\ldots,k-1\}$, the algorithm tries to detect a $K_k$ in $H$;
see also~\cite{KL22}. We give a few applications at the end of this section.

Here we focus on $k=3$. For convenience, define the following functions on integer inputs
\begin{align}
 \beta(\ell) &= \omega(\lfloor \ell/3 \rfloor, \lceil (\ell-1)/3 \rceil,  \lceil \ell/3 \rceil), \label{eq:beta}  \\
 \gamma(\ell) &= \lfloor \ell/3 \rfloor \omega + \ell \pmod 3. \label{eq:gamma}
\end{align}

With this notation, the algorithm runs in $O(n^{\gamma(\ell)})$ time.
Two decades later, Eisenbrand and Grandoni~\cite{EG04} refined the triangle method by using fast algorithms
for rectangular matrix multiplication instead of those for square matrix multiplication.
It partitions the graph into three parts of roughly the same size: $\ell_1=\lfloor \ell/3 \rfloor$,
$\ell_2=\lceil (\ell-1)/3 \rceil$, $\ell_3 = \lceil \ell/3 \rceil$.  
The refined algorithm runs in time  $O(n^{\beta(\ell)})$ time.
If the rectangular matrix multiplication is carried out via the straightforward partition into square blocks
and fast square matrix multiplication (see, \eg, \cite[Exercise~4.2-6]{CLRS09}), one recovers the
time complexity of the algorithm of Ne{\v{s}}et{\v{r}}il and Poljak; that is,
$\beta(\ell) \leq \gamma(\ell)$, see~\cite{EG04} for details. 
Eisenbrand and Grandoni~\cite{EG04} showed that the above inequality is strict for a certain range:
if $\ell \pmod 3 =1$ and $\ell \leq 16$, or $\ell \pmod 3=2$.
In summary, their algorithm is faster than that of Ne{\v{s}}et{\v{r}}il and Poljak in these cases.
We subsequently refer to their method as the \emph{refined  triangle method}. 
Another extension of the triangle method is considered in~\cite{KL22}. 

\paragraph{A problem of Kloks, Kratsch, and M\"{u}ller.} The authors asked~\cite{KKM00}
whether there is an $O(m^{\ell \omega/6})$ algorithm for deciding whether a graph contains a $K_\ell$,
if $\ell$ is a multiple of $3$. Eisenbrand and Grandoni showed that this is true for every multiple of $3$
at least $6$. Indeed, by their Theorem 2~\cite{EG04}, this task can be accomplished in time
$O(m^{\beta(\ell)/2})$ for every $\ell \geq 6$, with $\beta(\ell)$ as in~\eqref{eq:beta}.
If $\ell =3j$, where $j \geq 2$, then
\[ \beta(\ell) =  \omega(\lfloor \ell/3 \rfloor, \lceil (\ell-1)/3 \rceil,  \lceil \ell/3 \rceil) =
\omega(j,j,j) = j \cdot \omega(1,1,1)= j \omega. \]
It follows that the running time is $O(m^{\beta(\ell)/2}) = O(m^{j \omega /2}) = O(m^{\ell \omega /6})$. 
The proof of the theorem is rather involved. Here we provide an alternative simpler argument
that also yields an arboricity-sensitive bound, item (i) below.      

\paragraph{General case derivations.}
We first consider the general cases: (i)~$\ell=3j$, $j \geq 3$;
(ii)~$\ell=3j+1$, $j \geq 2$; and (iii)~$\ell=3j+2$, $j \geq 2$.
It will be evident that our algorithms provide improved bounds for every $\ell \geq 7$. 
For a given $j \geq 2$, we shall consider the interval
\[ I_j =\left(\frac{\omega-1}{j(3 -\omega)  + 2(\omega-1)}, \, \frac12 \right). \]

\begin{enumerate} [(i)] \itemsep 1.5pt

\item $\ell=3j$, $j \geq 3$. We use the triangle method with a refined calculation.
The vertices of the auxiliary graph $H$ are
subgraphs isomorphic to $K_j$. By the result of Chiba and Nishizeki~\cite{CN85}, $H$ has 
$O(\alpha^{j-2} \cdot m)$ vertices. The algorithm counts triangles in $H$ in time proportional to 
\[ \left(\alpha^{j-2} \cdot m \right)^\omega = \alpha^{(j-2) \omega} \cdot m^\omega. \]
For $\ell=9,12,15$ (the entries in Table~\ref{tab:summary}), these runtimes are
$O(\alpha^\omega \cdot m^\omega)$, $O(\alpha^{2\omega} \cdot m^\omega)$, and $O(\alpha^{3\omega} \cdot m^\omega)$, 
respectively.

Since $\alpha = O(m^{1/2})$, the above expression is bounded from above as follows:
\[ \alpha^{(j-2) \omega} \cdot m^\omega =O\left( \left(m^{(j-2)/2} \cdot m\right)^\omega \right)
= O\left(m ^{j \omega/2}\right) = O \left(m^{\ell \omega/6} \right). \]

Next, we show that for a certain high-range of $\alpha$, the new bound
$\alpha^{(j-2) \omega} \cdot m^\omega$ beats all previous bounds,
namely $m^{j \omega/2}$ and $\alpha^{3j-2} \cdot m$, for $j \geq 3$ (or $\ell =9,12,15,\ldots$). 
Let $\alpha=\Theta(m^x)$, where $x \in I_j$. 
We first verify that
\begin{align*}
\alpha^{(j-2) \omega} \cdot m^\omega \ll m^{j \omega/2},   \text{ or } 
((j-2)x+1)\omega < j \omega/2,
\end{align*}
\later{
\begin{align*}
\alpha^{(j-2) \omega} \cdot m^\omega &\ll m^{j \omega/2},   \text{ or } \\
((j-2)x+1)\omega &< j \omega/2,
\end{align*}
}
which holds for $x<1/2$.
Secondly, we verify that
  \begin{align*}
  \alpha^{(j-2) \omega} \cdot m^\omega &\ll \alpha^{3j-2} \cdot m, \text{ or } 
  ((j-2)x+1)\omega < (3j-2)x+1, \text{ or } \\
  x&> \frac{\omega-1}{j(3 -\omega)  + 2(\omega-1)},
  \end{align*}
  \later{
  \begin{align*}
  \alpha^{(j-2) \omega} \cdot m^\omega &\ll \alpha^{3j-2} \cdot m, \text{ or } \\
  ((j-2)x+1)\omega &< (3j-2)x+1, \text{ or } \\
  x&> \frac{\omega-1}{j(3 -\omega)  + 2(\omega-1)},
  \end{align*}
  }
which holds by the choice of the interval $I_j$. 
In particular, for $\ell=9$, this occurs for $x \in \left(\frac{\omega-1}{7 -\omega}, \frac12 \right) \supset (0.297,0.5)$;
for $\ell=12$, this occurs for $x \in \left(\frac{\omega-1}{10 -2\omega}, \frac12 \right) \supset (0.262,0.5)$;
for $\ell=15$, this occurs for $x \in \left(\frac{\omega-1}{13 -3\omega}, \frac12 \right) \supset (0.234,0.5)$.
Moreover, if $\omega=2$, as conjectured, these intervals extend to $(1/5,1/2)$, $(1/6,1/2)$, and
$(1/7,1/2)$, respectively. 

\item $\ell=3j+1$, $j \geq 2$. The refined triangle method~\cite{EG04} with $\ell_1=j$, $\ell_2=j$, $\ell_3=j+1$,
  leads to rectangular matrix multiplication
$[O(\alpha^{j-2} m) \times O(\alpha^{j-2} m)] \cdot [O(\alpha^{j-2} m) \times O(\alpha^{j-1} m)]$.
Its complexity is at most $O(\alpha)$ times that of the square matrix multiplication with dimension $\alpha^{j-2} m$,
the latter of which is $O(\alpha^{(j-2)\omega} \cdot m^\omega)$. It follows that
\[ T(n,m,\ell) = O(\alpha \cdot \alpha ^{(j-2)\omega} \cdot m^\omega) = O(\alpha^{(j-2) \omega+1} \cdot m^\omega). \] 
For $\ell=7,10,13,16$ (the entries in Table~\ref{tab:summary}), these runtimes are
$O(\alpha \cdot m^\omega)$, $O(\alpha^{\omega +1} \cdot m^\omega)$, $O(\alpha^{2\omega + 1} \cdot m^\omega)$,
and $O(\alpha^{3\omega + 1} \cdot m^\omega)$, respectively.

As before, we show that for a certain high-range of $\alpha$, the new bound
$\alpha^{(j-2) \omega +1} \cdot m^\omega$ beats all previous bounds,
namely $m^{(j \omega+1)/2}$ and $\alpha^{3j-1} \cdot m$, for $j \geq 2$ (or $\ell =7,10,13,\ldots$). 
Let $\alpha=\Theta(m^x)$, where $x \in I_j$. 
We first verify that
  \begin{align*}
\alpha^{(j-2) \omega +1} \cdot m^\omega \ll m^{(j \omega+1)/2},   \text{ or } 
((j-2)x+1)\omega < j \omega/2, 
  \end{align*}
\later{
  \begin{align*}
\alpha^{(j-2) \omega +1} \cdot m^\omega &\ll m^{(j \omega+1)/2},   \text{ or } \\
((j-2)x+1)\omega &< j \omega/2, 
  \end{align*}
} 
which holds for $x<1/2$.
Secondly, we verify that
\begin{align*}
  \alpha^{(j-2) \omega+1} \cdot m^\omega &\ll \alpha^{3j-1} \cdot m, \text{ or } %\\
  ((j-2)x+1)\omega +x < (3j-1)x+1, \text{ or } \\
  x&> \frac{\omega-1}{j(3 -\omega)  + 2(\omega-1)},
  \end{align*}
which holds by the choice of the interval $I_j$. 
In particular, for $\ell=7$, this occurs for $x \in \left(\frac{\omega-1}{4}, \frac12 \right) \supset (0.344,0.5)$;
for $\ell=10$, this occurs for $x \in \left(\frac{\omega-1}{7 -\omega}, \frac12 \right) \supset (0.297,0.5)$;
for $\ell=13$, this occurs for $x \in \left(\frac{\omega-1}{10 -2\omega}, \frac12 \right) \supset (0.262,0.5)$.
Moreover, if $\omega=2$, as conjectured, these intervals extend to $(1/4,1/2)$, $(1/5,1/2)$, and
$(1/6,1/2)$, respectively. 

\item $\ell=3j+2$, $j \geq 2$. The refined triangle method~\cite{EG04} with $\ell_1=j$, $\ell_2=j+1$, $\ell_3=j+1$,
  leads to rectangular matrix multiplication
$[O(\alpha^{j-2} m) \times O(\alpha^{j-1} m)] \cdot [O(\alpha^{j-1} m) \times O(\alpha^{j-1} m)]$.
Its complexity is at most $O(\alpha^2)$ times that of the square matrix multiplication with dimension $\alpha^{j-2} m$,
the latter of which is $O(\alpha^{(j-2)\omega} \cdot m^\omega)$. It follows that
\[ T(n,m,\ell) = O(\alpha^2 \cdot \alpha ^{(j-2)\omega} \cdot m^\omega) = O(\alpha^{(j-2) \omega+2} \cdot m^\omega). \] 
For $\ell=8,11,14$ (the entries in Table~\ref{tab:summary}), these runtimes are
$O(\alpha^2 \cdot m^\omega)$, $O(\alpha^{\omega +2} \cdot m^\omega)$, and $O(\alpha^{2\omega + 2} \cdot m^\omega)$, 
respectively. 

As before, we show that  for a certain high-range of $\alpha$, this bound beats all previous bounds,
namely $m^{(j \omega+2)/2}$ and $\alpha^{3j} \cdot m$, for $j \geq 2$ (or $\ell =8,11,14,\ldots$). 
Let $\alpha=\Theta(m^x)$, where $x \in I_j$. 
We first verify that
\begin{align*}
\alpha^{(j-2) \omega +2} \cdot m^\omega \ll m^{(j \omega+2)/2},   \text{ or } %\\
((j-2)x+1)\omega < j \omega/2,
\end{align*}
which holds for $x<1/2$.
Secondly, we verify that
\begin{align*}
  \alpha^{(j-2) \omega+2} \cdot m^\omega &\ll \alpha^{3j} \cdot m, \text{ or } %\\
  ((j-2)x+1)\omega +2x < 3jx+1, \text{ or } \\
  x&> \frac{\omega-1}{j(3 -\omega)  + 2(\omega-1)},
  \end{align*}
which holds by the choice of the interval $I_j$. 
In particular, for $\ell=8$, this occurs for $x \in \left(\frac{\omega-1}{4}, \frac12 \right) \supset (0.344,0.5)$;
for $\ell=11$, this occurs for $x \in \left(\frac{\omega-1}{7 -\omega}, \frac12 \right) \supset (0.297,0.5)$;
for $\ell=14$, this occurs for $x \in \left(\frac{\omega-1}{10 -2\omega}, \frac12 \right) \supset (0.262,0.5)$.
Moreover, if $\omega=2$, as conjectured, these intervals extend to $(1/4,1/2)$, $(1/5,1/2)$, and
$(1/6,1/2)$, respectively.

The bound for $\ell=8$, $O(\alpha^2 \cdot m^\omega)$, is superseded by the bound
$O\left(\alpha^{\omega+1} m^{\frac{\omega+1}{2}} \right)$, using another (so-called ``the edge count'') method. 

\end{enumerate}

\paragraph{Running time derivations for small $\ell$.}
The previous general case derivations together with the instantiations below yield the running times listed in
Table~\ref{tab:summary}.

\smallskip
\begin{table}[htbp]
  \centering
  %{\small
  \begin{tabular}{|c|c|c|} \toprule
    $\ell$ & Previous best & This paper \\ \midrule
    $3$ & $O( \alpha \cdot m)$ \cite{CN85}, $O(n^{2.372})$ \cite{IR78}, $O(m^{1.407})$ \cite{AYZ97} &  \\
    $4$ & $O( \alpha^2 \cdot m)$ \cite{CN85}, $O(n^{3.334})$ \cite{EG04}, $O(m^{1.686})$ \cite{KKM00} & $O(n^{3.251})$ \\
    $5$ & $O( \alpha^3 \cdot m)$ \cite{CN85}, $O(n^{4.220})$, $O(m^{2.147})$ \cite{EG04} & $O(n^{4.086})$ \\
    $6$ & $O( \alpha^4 \cdot m)$ \cite{CN85}, $O(n^{4.751})$ \cite{EG04}, $O(m^{2.372})$ \cite{EG04} & \\
    $7$ & $O( \alpha^5 \cdot m)$ \cite{CN85}, $O(n^{5.714})$ \cite{EG04}, $O(m^{2.857})$ \cite{EG04} & $O(m^{2.795})$, $O(\alpha \cdot m^{2.372})$ \\
    $8$ & $O( \alpha^6 \cdot m)$ \cite{CN85}, $O(m^{3.372})$ \cite{EG04,KKM00} & $O(m^{3.199})$, $O(\alpha^{3.372} \cdot m^{1.686})$ \\ %%% $O(\alpha^2 \cdot m^{2.372})$ \\ 
    $9$ &  $O( \alpha^7 \cdot m)$ \cite{CN85}, $O(m^{3.558})$ \cite{EG04,KKM00} & $O(\alpha^{2.372} \cdot m^{2.372})$ \\
    $10$ &  $O( \alpha^8 \cdot m)$ \cite{CN85}, $O(m^{4.058})$ \cite{EG04,KKM00} & $O(m^{3.976})$, $O(\alpha^{3.372} \cdot m^{2.372})$ \\
    $11$ &  $O( \alpha^9 \cdot m)$ \cite{CN85}, $O(m^{4.558})$  \cite{EG04,KKM00} & $O(m^{4.372})$, $O(\alpha^{4.372} \cdot m^{2.372})$ \\
    $12$ & $O( \alpha^{10} \cdot m)$ \cite{CN85}, $O(m^{4.744})$ \cite{EG04,KKM00} & $O(\alpha^{4.744}  \cdot m^{2.372})$, $O(\alpha^{6.744} \cdot m^{1.686})$ \\
    $13$ & $O( \alpha^{11} \cdot m)$ \cite{CN85}, $O(m^{5.160})$ \cite{EG04,KKM00} & $O(\alpha^{5.744}  \cdot m^{2.372})$ \\
    $14$ & $O( \alpha^{11} \cdot m)$ \cite{CN85}, $O(m^{5.553})$ \cite{EG04,KKM00} & $O(m^{5.552})$, $O(\alpha^{6.744}  \cdot m^{2.372})$ \\
    $15$ & $O( \alpha^{13} \cdot m)$ \cite{CN85}, $O(m^{5.930})$ \cite{EG04,KKM00} & $O(\alpha^{7.116}  \cdot m^{2.372})$ \\ 
    $16$ & $O( \alpha^{14} \cdot m)$ \cite{CN85}, $O(m^{6.338})$ \cite{EG04,KKM00} & $O(\alpha^{8.115}  \cdot m^{2.372})$, $O(\alpha^{10.115} \cdot m^{1.686})$ \\  \bottomrule
  \end{tabular}
  %}
\smallskip
  \caption{Running time comparison for finding/counting complete subgraphs.
    The column on new results includes bounds in terms of $m$,  based on up-to-date matrix multiplication times, 
    and new arboricity-sensitive bounds in terms of $m$ and $\alpha$; the entries for $\ell=4,5$ in terms of $n$ 
    are also derived in~\cite{KL22}.}
\label{tab:summary}
\end{table}

\begin{enumerate} [(i)] \itemsep 1.5pt

\item $\ell=3$. The algorithms of Itai and Rodeh~\cite{IR78}, and Alon, Yuster, and Zwick~\cite{AYZ97}
  for triangle counting/detection run in $O(n^\omega)$ and $O(m^{2\omega/(\omega+1)}) =O(m^{1.407})$ time,
  respectively.

\item $\ell=4$. The refined triangle method~\cite{EG04} with $\ell_1=1$, $\ell_2=1$, $\ell_3=2$, leads to
  rectangular matrix multiplication $[n \times n] \cdot [n \times m]$ or
  $[n \times n] \cdot [n \times n^2]$ in the worst case. 
  Since $\omega(1,1,2) = \omega(1,2,1) \leq 3.251$,
  it follows that $T(n,m,4) = O(n^{3.251})$. (See also~\cite[Table~3]{GU18}; 
  this entry has been slightly improved in~\cite{VXXZ24}.)
  In particular, this gives a positive answer to a question of Alon, Yuster, and Zwick~\cite{AYZ97}, who asked
  for an improvement of their $O(n^{\omega+1})= O(n^{3.372})$ upper bound for $K_4$ detection. 
  In terms of $m$, by~\cite[Table~1]{EG04}, $T(n,m,4) = O(m^{1.682})$, but this entry appears unjustified. 

\item $\ell=5$. The refined triangle method~\cite{EG04} with $\ell_1=1$, $\ell_2=2$, $\ell_3=2$, leads to
  rectangular matrix multiplication $[n \times m] \cdot [m \times m]$ or $[n \times n^2] \cdot [n^2 \times n^2]$ in the
  worst case. According to~\cite[Table~1]{VXXZ24}, $\omega(1,2,2) = 2 \omega(1,0.5,1) \leq 2 \cdot 2.043 = 4.086$,
  hence $T(n,m,5) = O(n^{4.086})$.

On a side note, the extension method~\cite{EG04} with $j_1=3, j_2=2$, yields
$T(n,m,5) = O( \alpha m \cdot m) = O(\alpha \cdot m^2)$. 
Observe however, that $O(\alpha^3 \cdot m) = O(\alpha \cdot m^2)$, so the above
upper bound does not beat that of the algorithm of Chiba and Nishizeki.

\item $\ell=6$. The general case derivation yields a runtime of  $O(m^{\omega})$. 

\item $\ell=7$. The refined triangle method~\cite{EG04} with $\ell_1=2$, $\ell_2=2$, $\ell_3=3$, leads to
  rectangular matrix multiplication $[m \times m] \cdot [m \times \alpha m]$ or
  $[m \times m] \cdot [m \times m^{3/2}]$ in the worst case. 
  Since $\omega(1,1.5,1) \leq 2.795$, it follows that $T(n,m,7) = O(m^{2.795})$.

\item $\ell=8$.
  By~\cite[Thm.~2]{EG04}, it takes $O(m^{\beta(8)/2})$ time, where $\beta(8)=\omega(2,3,3)$.
  According to~\cite[Table~1]{ADW+25},
  $\omega(1,0.6,1)\leq 2.0924$ and $\omega(1,0.7,1)\leq 2.1528$.
  By the convexity of $\omega(x,y,z)$ (see~\cite{LR83}),
  we have $\beta(8)=\omega(3, 2, 3) \leq \omega(1,0.6,1)+2\omega(1,0.7,1) \leq 6.398$.
  Thus, $T(n,m,8) = O(m^{\beta(8)/2})=O(m^{3.199})$. 

\item $\ell=9$. The general case derivation yields a runtime of $O(\alpha^\omega \cdot m^\omega)$.

\item $\ell=10$.  By~\cite[Thm.~2]{EG04}, it takes $O(m^{\beta(10)/2})$ time, 
  where $\beta(10) = \omega(3,3,4)$.
  Since according to~\cite[Table~3]{GU18},
  $\omega(1,1.3,1) \leq 2.6217$ and $\omega(1,1.4,1)\leq 2.7084$,
  by the convexity of $\omega(x, y, z)$, we have
  $\omega(3,3,4)=\omega(3,4,3) = 3\omega(1, 1.3*2/3+1.4/3, 1) \leq 2\omega(1,1.3,1)+\omega(1,1.4,1) \leq 7.9518$.
  Thus, $T(n,m,10)=O(m^{\beta(10)/2}) = O(m^{3.976})$.
Example: for graphs with $\alpha=O(m^{3/8})$, $K_{10}$'s can be counted/detected in $O(m^{3.636})$ time,
by the last entry in the right column, whereas the other entries give $O(m^{3.976})$ time or more. 

 \item $\ell=11$.  By~\cite[Thm.~2]{EG04}, it takes $O(m^{\beta(11)/2})$ time, 
  where $\beta(11) = \omega(3,4,4)$.
 By~\cite[Table~1]{VXXZ24}, we have $\omega(1,0.75,1) \leq 2.1863$, which yields 
 $\omega(3,4,4) = 4 \omega(1,0.75,1) \leq 4 \times 2.1863 \leq 8.746$.
 Consequently, $T(n,m,11) = O(m^{\beta(11)/2}) = O(m^{4.373})$.

\item $\ell=12$. The general case derivation yields a runtime of $O(\alpha^{2\omega} \cdot m^\omega)$. 
  Example: for graphs with $\alpha=O(m^{0.4})$, $K_{12}$'s can be counted/detected in 
  $O(\alpha^{2\omega} \cdot m^\omega) = O(m^{1.8 \omega}) = O(m^{4.269})$ time, whereas the other entries
  give $O(m^{4.744})$ time or more. 

\item $\ell=13$.  By~\cite[Thm.~2]{EG04}, it takes $O(m^{\beta(13)/2})$ time, 
  where $\beta(13) = \omega(4,4,5)$. 
  By~\cite[Table~3]{GU18}, we have $\omega(1,5/4,1) \leq 2.58$, which yields 
  $\omega(4,4,5) = 4 \omega(1,5/4,1) \leq 4 \times 2.58=10.32$.
  Consequently, $T(n,m,13) = O(m^{\beta(13)/2}) = O(m^{5.16})$. 

\item $\ell=14$.  By~\cite[Thm.~2]{EG04}, it takes $O(m^{\beta(14)/2})$ time, 
  where $\beta(14) = \omega(4,5,5) = \omega(5,4,5)$. 
  By~\cite[Table~1]{ADW+25}, we have $\omega(1,0.8,1) \leq 2.2207$, which implies
  $\omega(5,4,5) = 5 \omega(1,0.8,1) \leq 5 \times 2.2207=11.1035$,
  thus  $T(n,m,14) = O(m^{\beta(14)/2}) = O(m^{5.552})$. 

\item $\ell=15$. The general case derivation yields a runtime of  $O(\alpha^{3\omega} \cdot m^\omega)$.

\item $\ell=16$.  By~\cite[Thm.~2]{EG04}, it takes $O(m^{\beta(16)/2})$ time, 
  where $\beta(16) = \omega(5,5,6)$. 
  By~\cite[Table~1]{VXXZ24}, we have $\omega(1,6/5,1) \leq 2.5351$, which implies
  $\omega(5,5,6) = 5 \omega(1,6/5,1) \leq 5 \times 2.5351 \leq 12.676$,
  thus  $T(n,m,16) = O(m^{\beta(16)/2}) = O(m^{6.338})$. 

\end{enumerate}

\paragraph{The edge count method.} This method combines the extended triangle method with
runtime upper bounds of algorithms for the ``sparse'' case of the detection problem.
We illustrate the method for $\ell=4 j$, with $j \geq 2$.
Detecting the presence of a $K_\ell$ amounts to detecting the presence of a $K_4$
in a graph $H$ with $O(\alpha^{j-2} m)$ vertices corresponding to the $K_j$'s in $G$, where two vertices in $H$
are adjacent if they span $2j$ vertices of $G$ that form a $K_{2j}$. There are $k=O(\alpha^{2j -2} m)$ $K_{2j}$'s in $G$,
and so by the result of~\cite{KKM00}, a $K_4$ in $H$ can be detected in
$O(k^{(\omega+1)/2})=O\left( \alpha^{(j-1) (\omega+1)} \cdot m^{\frac{\omega+1}{2}} \right)$ time.

Consider three examples: $j=2,3,4$.
For $j=2$, the new bound  $ O\left( \alpha^{\omega+1} \cdot m^{\frac{\omega+1}{2}} \right)$ is better
than the earlier bound:. Indeed, it is easily seen that
$ \alpha^{\omega+1} \cdot m^{\frac{\omega+1}{2}} = O(\alpha^2 m ^\omega)$
over the entire range of $\alpha$. 
For $j=3,4$, the new bounds are better than the earlier bounds for $\alpha \ll m^{\frac{\omega -1}{4}}$.
Indeed,  $\alpha^{2(\omega+1)} \cdot m^{\frac{\omega+1}{2}} \ll \alpha^{2 \omega} \cdot m^\omega$
and $\alpha^{3(\omega+1)} \cdot m^{\frac{\omega+1}{2}} \ll \alpha^{3\omega + 1} \cdot m^\omega$
for $\alpha \ll m^{\frac{\omega -1}{4}}$. 

These examples are not exhaustive. For instance, one can also apply the edge count method with $\ell= 3 j$
and use the $O(m^{2\omega/(\omega+1)}) =O(m^{1.407})$ time algorithm for triangle detection~\cite{AYZ97}.

\section{Conclusion}   \label{sec:conclusion}

We conclude by restating a few open problems in the area of finding small complete subgraphs.
Recall  that triangle detection can be done in $O(n^\omega)$ time by the algorithm of Itai and Rodeh.

\begin{enumerate} \itemsep 1pt
 
\item Can $K_3$ detection be done in $O(n^2)$ time?

\item Is $K_3$ detection as difficult as Boolean matrix multiplication?
 Initially posed by Alon, Yuster, and Zwick~\cite{AYZ95}, it also appears as 
 in Woeginger's survey on exact algorithms~\cite{Woe08} (as Question 4.3\,(c)). 
V. V. Williams and R. Williams showed that if any of the two problems admits a truly subcubic
\emph{combinatorial} algorithm then also the other one does~\cite{VW18}.

\item Can $K_{10}$ detection be accomplished in $O(n^{7.5})$ time?
  (Question 4.3\,(b) in Woeginger's survey on exact algorithms~\cite{Woe08})
  We add our own version: Can $K_{100}$ detection be accomplished in $O(n^{75})$ time?

\end{enumerate}

\paragraph{Acknowledgment.}
The authors thank an anonymous reviewer for communicating the short folklore proof of Lemma~\ref{lem:F}.

\end{document}